\documentclass[draftcls,onecolumn,twoside]{IEEEtran}
\usepackage{graphics,graphicx,subfigure}
\usepackage{stfloats}
\usepackage{amsfonts,amssymb,amsmath,amsbsy,bm,paralist,theorem,ifthen,color}
\usepackage{calc}
\usepackage{hyperref}
\usepackage{array,multirow}
\usepackage{algorithmic,algorithm}
\usepackage{enumerate}
\usepackage{cases}
\usepackage[noadjust]{cite}
\hypersetup{unicode,%
            bookmarksnumbered,%
            bookmarksopen=true,
            bookmarksopenlevel=2,
            breaklinks=true,
            colorlinks=true,
            linkcolor=blue,
            citecolor=blue,
            filecolor=black,
            urlcolor=black,
            plainpages=false,%
            pdfborder={0 0 0},
            pdfstartview=FitH
           }

\newcommand{\Cset}   {\ensuremath{\mathbb{C}}}

\newcommand{\CN}     {\ensuremath{\mathcal{CN}}}

\newcommand{\hb}     {\ensuremath{\mathbf{h}}}
\newcommand{\vb}     {\ensuremath{\mathbf{v}}}

\newcommand{\xb}     {\ensuremath{\mathbf{x}}}
\newcommand{\wb}     {\ensuremath{\mathbf{w}}}

\newcommand{\zerob}  {\ensuremath{\mathbf{0}}}

\newcommand{\Vb}     {\ensuremath{\mathbf{V}}}

\newcommand{\Gb}     {\ensuremath{\mathbf{G}}}

\newcommand{\taub}   {\ensuremath{\bm{\tau}}}

\DeclareMathOperator*{\argmax}{arg\,max}

\DeclareMathOperator*{\st}{s.t.}

\DeclareMathOperator*{\Tr}{Tr}

\DeclareMathOperator*{\E}{\mathbb{E}}

\newtheorem{Remark}{Remark}
\newtheorem{lemma}[Remark]{Lemma}

\begin{document}
\title{Joint Beamforming Design and Time Allocation for Wireless Powered Communication Networks}
\author{Qian Sun, Gang Zhu, Chao Shen,~\IEEEmembership{Member,~IEEE,} Xuan Li and Zhangdui Zhong
\thanks{The authors are with the State Key Laboratory of Rail Traffic Control and Safety, Beijing Jiaotong University, Beijing, China 100044.
Qian Sun and Chao Shen are also with the State Key Laboratory of Integrated Services Networks, Xidian University, Xi¡¯an, China
(e-mail: \{12120137, gzhu, shenchao, 12120099, zhdzhong\}@bjtu.edu.cn).
Chao Shen is the corresponding author.
The authors would like to thank Prof. Tsung-Hui Chang for his valuable suggestions.
}
}
\markboth{Submitted to IEEE Communications Letters, March 2014}{Submitted to IEEE Communications Letters, March 2014}
\maketitle
\begin{abstract}
This paper investigates a  multi-input single-output (MISO) wireless powered communication network (WPCN) under the protocol of harvest-then-transmit.
The power station (PS) with reliable power supply can replenish the passive user nodes by wireless power transfer (WPT) in the downlink (DL), then each user node transmits independent information to the sink by a time division multiple access (TDMA) scheme in the uplink (UL).
We consider the joint time allocation and beamforming design to maximize the system sum-throughput.
The semidefinite relaxation (SDR) technique is applied to solve the nonconvex design problem. The tightness of SDR approximation, thus the global optimality, is proved. This implies that only one single energy beamformer is required at the PS.
Then a fast semi-closed form solution is proposed by exploiting the inherent structure.
Simulation results demonstrate the efficiency of the proposed algorithms from the perspectives of time complexity and information throughput.
\end{abstract}
\section{Introduction}
Recently, great interest has been drawn in energy harvesting from the radio-frequency (RF) signals in energy constrained wireless systems.
Due to the large-scale fading of RF signals, it is of great importance to carefully design the systems to improve the efficiency and performance by, e.g., the multiple-antenna beamforming.

Since RF signals can carry energy and information at the same time, lots of research works have emerged on the topic of {\it simultaneous wireless information and power transfer} (SWIPT) in recent years, e.g., \cite{Varshney,Grover,Ho,Lidd2013,ChaoShen2013,Lee2014a}, and references therein. The power allocation and/or beamforming designs were investigated under the point-to-point, broadcasting or relay channels with single-input single-output (SISO), multiple-input single-output (MISO) or multiple-input multiple-output (MIMO) configurations. The tradeoff between the energy signal and interference signal is the key point of SWIPT systems.

Besides SWIPT, there is another research topic on energy harvesting, referred to as {\it wireless powered communication network} (WPCN). In WPCN, the wireless systems are activated by the energy  via wireless power transfer (WPT), and then the wireless system is operated for wireless information transfer (WIT).
A harvest-then-transmit scheme was studied in a single-antenna scenario \cite{Ju}, where the sensors harvest energy broadcasted by the power station (PS) in the downlink (DL) and then send their independent information to the sink in the uplink (UL).
The performance of WPCN system was also analyzed from the network perspective \cite{Lee2012}.

In this letter, we consider a MISO WPCN scenario with the harvest-then-transmit protocol, and investigate the joint beamforming design and time allocation to maximize the sum-throughput.
Firstly, the semidefinite relaxation (SDR) technique is used to resolve the nonconvexity issue. The tightness of SDR approximation is proved, which implies that the global optimality can be achieved and only one energy beamformer is required for WPT. Then, a fast semi-closed form solution is developed by fully exploiting the inherent structure, which can reduce the implementation complexity significantly.

This rest of the paper is organized as follows.
The system model and the formulation of the joint time allocation and beamforming design are presented in Sec. \ref{Sec2}.
Then in Sec. \ref{Sec3}, we propose a fast semi-closed form solution to attain the global optimality with very low complexity.
Sec. \ref{Sec5} presents the simulation results, and Sec. \ref{Sec6} concludes the paper.
\section{System Model and Problem Formulation\label{Sec2}}
\begin{figure}[!t]
  \centering
  \includegraphics[width=.5\linewidth]{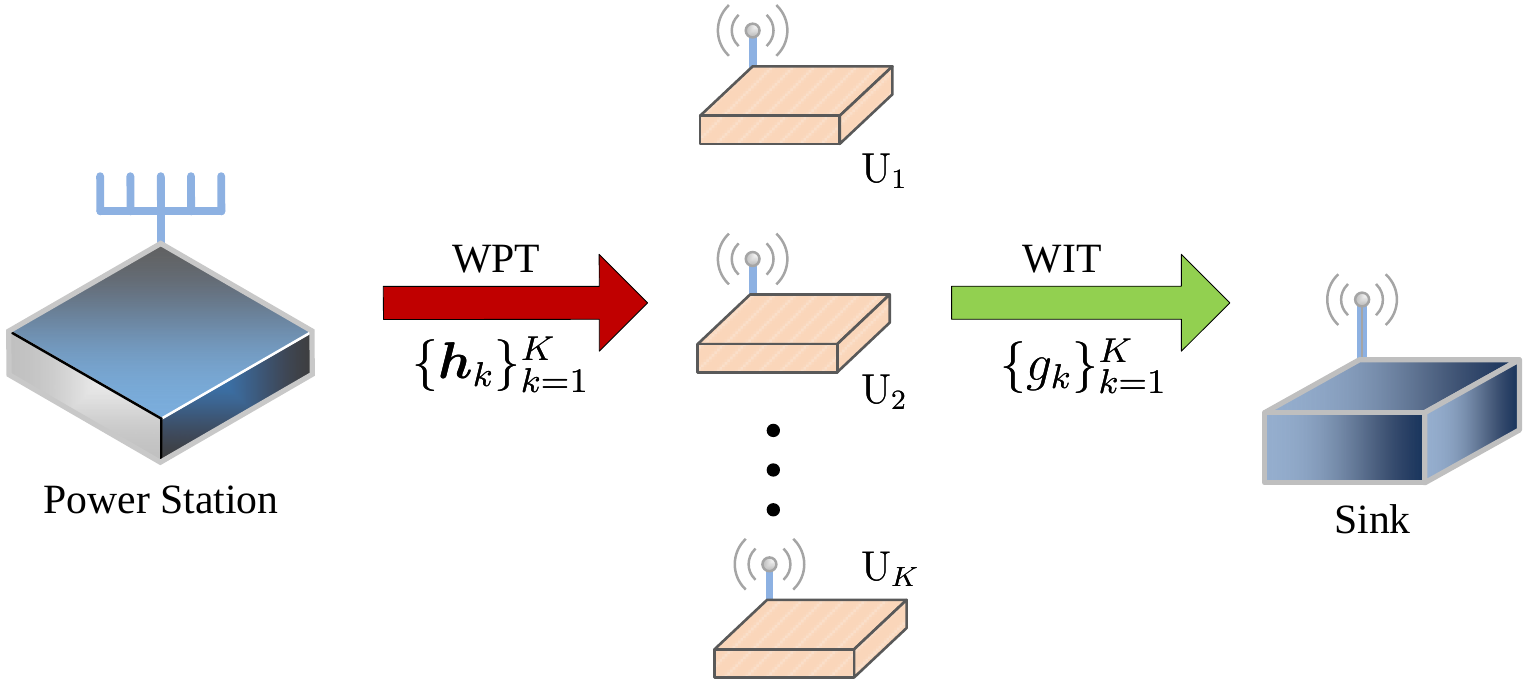}
  \caption{A MISO wireless powered communication network with a $N_t$-antenna power station, $K$ passive single-antenna user nodes and a sink node.}\label{fig:1}
\end{figure}

Consider a MISO WPCN illustrated in Fig. \ref{fig:1}, which consists of an $N_t$-antenna power station, a single-antenna sink and $K$ single-antenna user nodes, denoted by U$_k$ for $k\!=\!1,\ldots,K$. The network operates in a time division multiple access (TDMA) fashion. Assume that the frame duration is normalized to be unit. At the first $\tau_0\in[0,1]$ fraction of time, the PS broadcasts  power wirelessly to the $K$ user nodes in the DL. Then in the UL each user, say, U$_k$ for all $k$, sends its independent information to the sink one by one with $\tau_k\in[0,1]$ fraction of time, and by the energy harvested at the initial slot. The total time constraint reads $\sum_{k=0}^K \tau_k\leq 1$.

Suppose that the PS shall design $K$ WPT beamformers $\{\wb_k\}_{k=1}^K$ for $K$ users. Thus the transmitted signal at the PS can be expressed as
\begin{align}
   \xb(t)&=\sum_{k=1}^{K}\wb_ks_k(t),\label{P02:a}
\end{align}
where $\wb_k\in\Cset^{N_t}$ denotes the energy beamforming vector, and $s_k(t)$ is the Gaussian WPT signal with $s_k(t)\sim\CN(0,1)$.
The transmission power of the power station is limited by $P_{\max}$, i.e.,
\begin{align}
  \sum_{k=1}^{K}||\wb_k||^2_2&\leq P_{\max}.\label{P03:a}
\end{align}

Assume that both DL and UL channels are quasi-static flat-fading, denoted by $\hb_k\in\Cset^{N_t}$ and $g_k\in\Cset$, respectively, for all $k$. Further assume that the PS has perfect knowledge of all channel state information (CSI).
In the scenarios of WPT, it is reasonable to assume that the noise power is negligible compared with the power signal. Hence, the user node can harvest the energy from this RF signal with the amount
\begin{align}
   E_k&\approx\xi_k\tau_0\E\left(|\hb_k^H\xb(t)|^2\right)\\
   &=\xi_k\tau_0\E\Big|\hb_k^H\sum_{i=1}^K\wb_is_i(t)\Big|^2\\
   &=\xi_k\tau_0\sum_{i=1}^K\left|\hb_k^H\wb_i\right|^2,\label{P04:a}
\end{align}
where $\xi_k\in(0,1)$ accounts for the energy harvesting efficiency at user node $k$, for $k=1,\cdots,K$, and the noise power is ignored.
Suppose that the passive user nodes are powered only by the energy harvested from WPT,
and all harvested energy is used for its information transmission. Therefore, its average transmit power $P_k$ within the $\tau_k$ fraction of time  is given by
\begin{align}
    P_k =\frac{E_k}{\tau_k} = \frac{\tau_0}{\tau_k}\xi_k\sum_{i=1}^K\left|\hb_k^H\wb_i\right|^2, ~\forall k.\label{P05:a}
\end{align}

Let $t_k\sim \CN(0,P_k)$ be the signal transmitted by U$_k$. Then the received signal at the sink node in the $k$th UL slot  can be written by
\begin{align}
    r_k =g_kt_k+ z_k, k = 1,\cdots, K,\label{P06:a}
\end{align}
where $z_k\sim\CN(0,{\sigma}_k^2)$ represents the Gaussian noise at the sink node.
Thereby, the achievable throughput of node $k$ in bits/second/Hz (bps/Hz) follows
\begin{subequations}
  \begin{align}\!\!
    R_k(\bm {\tau},\{\wb_k\})
        &=\tau_k\log_2\left(1+\frac{|g_k|^2P_k}{\Gamma\sigma_k^2}\right)\\
        &=\tau_k\log_2\bigg(1+\frac{\gamma_k\tau_0}{\tau_k}\sum_{i=1}^K\left|\hb_k^H\wb_i\right|^2\bigg),\forall k,\label{P07:a}
  \end{align}
\end{subequations}
where ${\bm {\tau}}\triangleq[\tau_0,\tau_1, \cdots, \tau_K]^T$, and $\Gamma$ represents the signal-to-noise ratio (SNR) gap from the additive white Gaussian noise (AWGN) channel capacity due to a practical modulation and coding scheme used.
In addition, $\gamma_k$ is given by
\begin{align}
    \gamma_k =\frac{\xi_k |g_k|^2}{\Gamma\sigma_k^2},~ k = 1,\cdots, K.\label{P08:a}
  \end{align}

In this letter we consider the criterion of sum-throughput maximization. And thus the joint time allocation and beamforming design can be formulated as
\begin{subequations}\label{IddP1}
  \begin{align}
    \max_{\taub,\{\wb_k\}_{k=1}^K}&\sum_{k = 1}^K \tau_k\log_2\left(1+\frac{\gamma_k\tau_0}{\tau_k}\sum_{i=1}^K\left|\hb_k^H\wb_i\right|^2\right)\label{P1:a}\\
        \st~&\tau_k\geq 0, k=0,\cdots, K,~~\sum_{k = 0}^K \tau_k\leq 1,\label{P1:c}\\
        &\sum_{k = 1}^K ||\wb_k||_2^2\leq P_{\max},\label{P1:d}
  \end{align}
\end{subequations}
which is nonconvex due to the coupling of $\{\tau_k\}$ and $\{\wb_k\}$ in the objective function. However, we will propose an algorithm to solve \eqref{IddP1} to the global optimum in the sequel.

\section{Fast Algorithm with Global Optimality\label{Sec3}}

\subsection{Convex Reformulation of \eqref{IddP1}}
In order to resolve the nonconvexity issue, let us first introduce a set of auxiliary variables $\{\vb_k\}_{k=1}^K$ with $\vb_k=\sqrt{\tau_0}\wb_k$, $\forall k$. Hence, \eqref{IddP1} can be reformulated as
\begin{subequations}\label{IddP1a}
  \begin{align}
    ~~~\max_{\taub,\{\vb_k\}}~~&\sum_{k = 1}^K \tau_k\log_2\left(1+\frac{\gamma_k}{\tau_k}\sum_{i=1}^K\left|\hb_k^H\vb_i\right|^2\right)\\[-2pt]
    \st~~&\sum_{k = 0}^K \tau_k\leq 1,~\tau_k\geq 0,~ k=0,\ldots, K, \label{P2:b}\\[-2pt]
        &\sum_{k = 1}^K ||\vb_k||_2^2\leq \tau_0P_{\max}.\label{P2:d}
  \end{align}
\end{subequations}

One can readily verify that $\tau_0>0$ holds true at the optimal solution, so the optimal $\wb_k$ can always be recovered by solving the problem \eqref{IddP1a}.

Notice that $\sum_{i=1}^K\left|\hb_k^H\vb_i\right|^2=\hb_k^H\left(\sum_{i=1}^K\vb_i\vb_i^H\right)\hb_k$, where $\sum_{i=1}^K\vb_i\vb_i^H$ is a positive semidefinite (PSD) matrix with its rank no greater than $K$.
By leveraging the idea of semidefinite relaxation (SDR) \cite{Luo}, we replace the PSD matrix $\sum_{i=1}^{K}\vb_i\vb_i^H$ with a general-rank matrix $\Vb\succeq\zerob$. Therefore, the nonconvex programming \eqref{IddP1a} can be relaxed to
\begin{subequations}\label{IddP1b}
  \begin{align}
   \max_{\taub,\,\Vb}~~&\sum_{k = 1}^K \tau_k\log_2\Big(1+\frac{\gamma_k}{\tau_k}\Tr\left(\hb_k\hb_k^H\Vb\right)\Big)\label{IddP1b:a}\\[-2pt]
    \st~~&\sum_{k = 0}^K \tau_k\leq 1,~\tau_k\geq 0, ~ k=0,\cdots, K,\label{P3a:c}\\[-1pt]
        &\Tr(\Vb)\leq \tau_0P_{\max},~\Vb\succeq\zerob.
  \end{align}
\end{subequations}

The constraints in \eqref{IddP1b} are linear while the objective function is the sum of perspective functions \cite{Boyd} of the concave function $\log_2\left(1+ \gamma_k\Tr(\hb_k\hb_k^H\Vb)\right)$ for $k=1,\ldots,K$. We then conclude that the objective function, and hence the problem \eqref{IddP1b}, are concave w.r.t. $\taub$ and $\Vb$. Consequently, the global optimal solution can be obtained by any off-the-shelf interior point solver, e.g. {\tt CVX} \cite{cvx}.

It is worthy pointing out that the optimal solution to \eqref{IddP1b}, denoted by $\{\taub^*, \Vb^*\}$,  is unique due to its strict concavity. However, it is not the case for the problem \eqref{IddP1a}. And moreover, from the view of implementation complexity, we remark that 
\begin{Remark}\label{Remark::1}
If $\Vb^*$ is rank-one, i.e., $\Vb^*=\vb^*{\vb^*}^H$, then an optimal and favorable solution of $\{\vb_k\}$ to the problem \eqref{IddP1a} is given by $\vb_k^*=\vb^*$ for $k=1$, and $\zerob$ otherwise. 
\end{Remark}

Remark 1 turns out that only a single energy beam is required for the optimal WPT  if ${\rm Rank}(\Vb^*)=1$.
Indeed, the rank-one optimality can be guaranteed.
\begin{lemma}\label{Lemma::1}
The optimal solution $\Vb^*$  to \eqref{IddP1b} is of rank one.
\end{lemma}
\begin{proof}
The optimal time allocation to the problem \eqref{IddP1b} exists. And for any give ${\taub}$, consider the following problem
   \begin{subequations}\label{IddP1c}
  \begin{align}
    ~~~\max_{\Vb}~~&\sum_{k = 1}^K \tau_k\log_2\Big(1+\frac{\gamma_k}{\tau_k}\Tr(\hb_k\hb_k^H\Vb)\Big)\\
    \st~~&\Tr(\Vb)\leq \tau_0P_{\max},~\Vb\succeq\zerob.
  \end{align}
\end{subequations}
The objective function of \eqref{IddP1c} is convex but nonlinear. Hence, by using the technique of {\it successive convex approximation} \cite{LiWeiChiang,meisam}, the optimal solution can be achieved by solving a series of linear programming, with the form stated below
   \begin{subequations}
  \begin{align}\!\!\!\!
    \Vb_{n+1}\triangleq\argmax_{\Vb}~&\sum_{k = 1}^K \frac{\gamma_k}{1\!+\!\frac{\gamma_k\Tr(\hb_k\hb_k^H\Vb_{n})}{\tau_k}}\Tr(\hb_k\hb_k^H\Vb)\!\\
    \st~&\Tr(\Vb)\leq \tau_0P_{\max}, ~\Vb\succeq\zerob,
  \end{align}
\end{subequations}
where $\Vb_n$ is the optimal solution at the $n$th iteration. Clearly, $\Vb_{n+1}$  is rank-one for all $n$ according to \cite[Lemma 3.1]{Huang}, which completes the proof of the rank-one optimality.
\end{proof}

\subsection{Fast Algorithm to \eqref{IddP1b}}
In view of the potential application scenarios of WPCN, we are particularly interested in developing a fast algorithm design to \eqref{IddP1b} with low complexity. To this end, we will fully exploit the inherent structure of \eqref{IddP1b} in this subsection.

First, it can be verified that the time should be used up at the optimal solution. Then, the optimal UL time allocation $\{\tau_k^*\}_{k=1}^K$ can be expressed as a function of the optimal $\tau_0^*$ and $\Vb^*$. To show this, let us rewrite \eqref{IddP1b} as
\begin{subequations}\label{Problem13}\begin{align}\!\!\!
    \max_{\tau_0\in[0,1],\Vb}&~f(\tau_0,\Vb)\triangleq\left\{
    \begin{array}{cl}\!\!\!\displaystyle
\max_{\{\tau_k\}_{k=1}^K}&\!\displaystyle\!\!\sum_{k = 1}^K \tau_k\log_2\left(\!1+\frac{\gamma_k}{\tau_k}\Tr(\hb_k\hb_k^H\Vb)\!\right)\label{P12:a}\\
  \st~&\!\!\!\tau_i\geq 0, ~\forall i=1,\cdots, K, \\
     &\!\!\!\sum_{i = 1}^K \tau_i= 1-\tau_0,
     \end{array}\!\!\!\right\}\\
     \st~&\Tr(\Vb)\leq \tau_0P_{\max},~\Vb\succeq \zerob.
\end{align}\end{subequations}

Due to the strict concavity of the objective function inside \eqref{P12:a} and thanks to the Jensen's inequality, for any given $\tau_0\in[0,1]$ the optimal $\{\tau_k^*\}_{k=1}^K$ is attained  if and only if
\begin{align}
  \frac{\gamma_k}{\tau_k^*}\Tr(\hb_k\hb_k^H\Vb)={\rm SNR},\forall k,~\sum_{k=1}^K\tau_k^* = 1-\tau_0,\label{P21:a}
  \end{align}
which yields
   \begin{subequations}\label{P22}\begin{align}
    {\rm SNR}&=\frac{\sum_{i=1}^K\gamma_i\Tr(\hb_i\hb_i^H\Vb)}{1-\tau_0},\\
    \tau_k^*&=\frac{(1-\tau_0)\gamma_k\Tr(\hb_k\hb_k^H\Vb)}{\sum_{i=1}^K\gamma_i\Tr(\hb_i\hb_i^H\Vb)}, ~\forall k.\label{P22:a}
  \end{align}\end{subequations}

By substituting (\ref{P22}) into $f(\tau_0,\Vb)$, \eqref{Problem13} reduces to  
\begin{subequations}\label{P3}
  \begin{align}\!\!
  \max_{{0\le\tau_0\leq 1}, \Vb}&~(1-\tau_0)\log_2\left(1+ \frac{\sum_{k=1}^{K}\gamma_k\Tr(\hb_k\hb_k^H\Vb)}{1-\tau_0}\right)\\
  \st~&~\Tr(\Vb)\leq \tau_0P_{\max},~\Vb\succeq \zerob,
  \end{align}
  \end{subequations}
which is equivalent to the convex programming
 \begin{subequations}\label{IddsignalCVX}
  \begin{align}
    \max_{{0\le\tau_0\leq 1}}\max_{\Vb}&~(1-\tau_0)\log_2\left(1+ \frac{\Tr(\Gb\Gb^H\Vb)}{1-\tau_0}\right)\label{P3:a}\\
        \st &~\Tr(\Vb)\leq \tau_0P_{\max},~\Vb\succeq\zerob,
  \end{align}
\end{subequations}
with $\Gb\triangleq[\sqrt{\gamma_1}{\hb}_1,\cdots,\sqrt{\gamma_K}{\hb}_K]$.

For any given $\tau_0$, the inner maximization problem of \eqref{IddsignalCVX} admits 
a closed-form optimal solution \cite[Proposition 2.1]{Ho}
\begin{align}
    \Vb^*=\tau_0P_{\max}{\bm\upsilon}{\bm\upsilon^H}, \label{Iddsignalsemi-closed:a}
\end{align}
where ${\bm\upsilon}$ is the principal eigenvector of $\Gb\Gb^H$.

\begin{algorithm}[!b]
  \caption{Proposed fast algorithm to the problem \eqref{IddP1b}}
  \label{Alg1}
  \begin{algorithmic}[1]
  \STATE  \textbf{Input:} $P_{\max}$ and $\{\hb_k,\gamma_k\}_{k=1}^K$;
  \STATE  Obtain ${\bm\upsilon}$ and $\lambda_{\max}$ by SVD of $\Gb\Gb^H$.
  \STATE  Obtain $\tau_0^*$ from \eqref{IddsignalnoCVX} by golden section search;
  \STATE  Obtain $\{\tau_k^*\}_{k=1}^K$ from \eqref{P22:a} with $\tau_0=\tau_0^*$;
  \STATE  Obtain $\wb_k^*=\frac{1}{\sqrt{\tau_0^*}}{\bm\upsilon}$, for $k=1$, and $\zerob$ otherwise;
  \STATE  {\bf Output:} $\{\tau_k^*\}_{k=0}^K$ and $\{\wb_k^*\}_{k=1}^K$.
  \end{algorithmic}
\end{algorithm}

Thus, by substituting (\ref{Iddsignalsemi-closed:a}) into the problem \eqref{IddsignalCVX}, it follows
\begin{align}
    \max_{{0\le\tau_0\leq 1}}&~(1-\tau_0)\log\Big(1+ \frac{\tau_0}{1-\tau_0}P_{\max}\lambda_{\max}\Big),\label{IddsignalnoCVX}
\end{align}
where $\lambda_{\max}$ is the principal eigenvalue of $\Gb\Gb^H$, and which can be efficiently solved by, e.g., the golden section search.

To summarize, we formalize the procedure of the proposed fast algorithm as Algorithm \ref{Alg1}.

\subsection{Deterministic Signalling for WPT\label{Sec4}}
In the WPT phase, $s_k(t)$ can be a deterministic power signal instead of the Gaussian input. It was shown in \cite{ChaoShen2013,Lee2014a} that the deterministic signalling can improve the performance of SWIPT systems since the interference caused by the WPT signal over the information signal can be cancelled. But the throughput cannot be improved for the system in this letter.

To show this, assume w.l.o.g. that $s_k(t)=1$. Then the harvested energy at user $k$ is given by $E_k^D\! = \!\xi_k\tau_0\big|\hb_k^H\sum_{i=1}^K\wb_i\big|^2$, $\forall k$.
The average transmit power $P_k^D$ and achievable throughput $R_k^D$ are then respectively given by
\begin{subequations}\begin{align}
    P_k^D&=\frac{ E_k^D}{\tau_k} = \frac{\tau_0}{\tau_k}\xi_k\bigg|\hb_k^H\sum_{i=1}^K\wb_i\bigg|^2,~\forall k,\label{P31:a}\\
    R_k^D&=\tau_k\log_2\Bigg(1+\gamma_k\bigg|\hb_k^H\sum_{i=1}^K\wb_i\bigg|^2\frac{\tau_0}{\tau_k}\Bigg),\forall k.
  \end{align}
\end{subequations}

Let $\bar\vb=\sqrt{\tau_0}\sum_{i=1}^K\wb_i$. Then the sum-throughput maximization problem reads
  \begin{subequations}\label{ConP1}
  \begin{align}
    \max_{\taub,\bar\vb}~~&\sum_{k = 1}^K \tau_k\log_2\Big(1+\frac{\gamma_k}{\tau_k}\left|\hb_k^H\bar\vb\right|^2\Big)\label{P4:a}\\
        \st~~&\tau_k\geq 0 ~\forall k,~\sum_{k = 0}^K \tau_k\leq 1,~\left\|\bar\vb\right\|^2\leq \tau_0 P_{\max}.\label{P4:d}
  \end{align}
\end{subequations}
By using the SDR, i.e., relaxing the rank-1 matrix $\bar\vb\bar\vb^H$ with  a general-rank PSD matrix $\bar\Vb$, \eqref{ConP1} can be approximated by a convex problem which is exactly the same to \eqref{IddP1b}.
\begin{Remark}
  In a WPCN system considered in this letter, the deterministic signalling cannot improve the system performance,
  but potentially help to reduce the implementation complexity.
\end{Remark}

\section{Numerical Results \label{Sec5}}
We consider a network as shown in Fig. \ref{fig:1}. The $K$ users are uniformly located in a line with total distance being 10 meters. The PS and sink are placed at the perpendicular bisector of the user array with the vertical distances being $d_p$ and $d_s$ meters, respectively. The PS is equipped with $N_t\!=\!4$ antennas.

The DL channel $\hb_k$ is modelled as
\begin{align}
\hb_k=\sqrt{\frac{K_R}{1+K_R}}\hb_k^{\rm LOS}+\sqrt{\frac{1}{1+K_R}}\hb_k^{\rm NLOS},
  \end{align}
where the Rician factor $K_R=3$, $\hb_k^{\rm NLOS}$ follows the standard Rayleigh fading, and $\hb_k^{\rm LOS}$ is the line of sight (LOS) with the form $\hb_k^{\rm LOS}\!=\![1,e^{j\alpha_k},\ldots,e^{j(N_t-1)\alpha_k}]^T$, $\alpha_k=-\pi\sin(\beta_k)$, and $\beta_k$ being the direction of U$_k$ to PS.
The average power of $\hb_k$ is then normalized by the path loss $10^{-3}{(d^{DL}_k)}^{-\alpha}$,
where $d^{\rm DL}_k$ is the distance between the PS and U$_k$, and $\alpha$ is the path loss exponent.
The UL channel $g_k$ follows i.i.d. Rayleigh fading. Specifically, $|g_k|^2=10^{-3}\rho^2_k{(d^{\rm UL}_k)}^{-\alpha}$, where $\rho_k$ follows the standard Rayleigh fading and $d^{\rm UL}_k$ is the distance between the sink and U$_k$.
The power limit $P_{\max}$ is set to be 30 dBm.
Let $\sigma_k^2=-70$ dBm, $\xi_k = 0.5$ for all $k$, and the SNR gap $\Gamma=9.8$ dB.
The simulation results are averaged over 1000 channel realizations.

\begin{figure}[!t]\centering
  \includegraphics[width=0.55\linewidth]{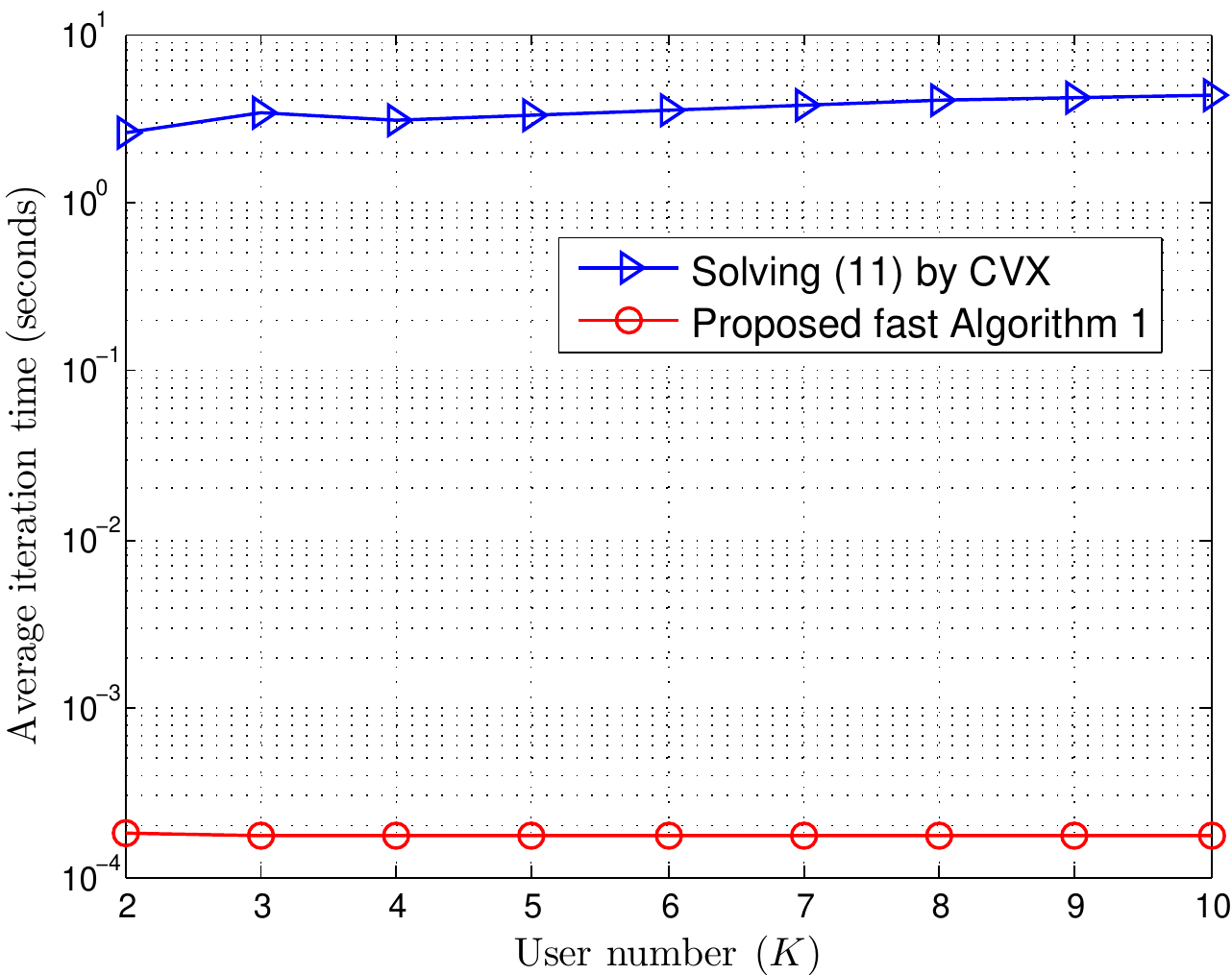}
  \caption{Average iteration time vs. user number, where $d_p\!=\!d_s\!=\!5$ meters.}\label{fig:IddTi}
\end{figure}
\begin{figure}[!t]\centering
   \includegraphics[width=0.55\linewidth]{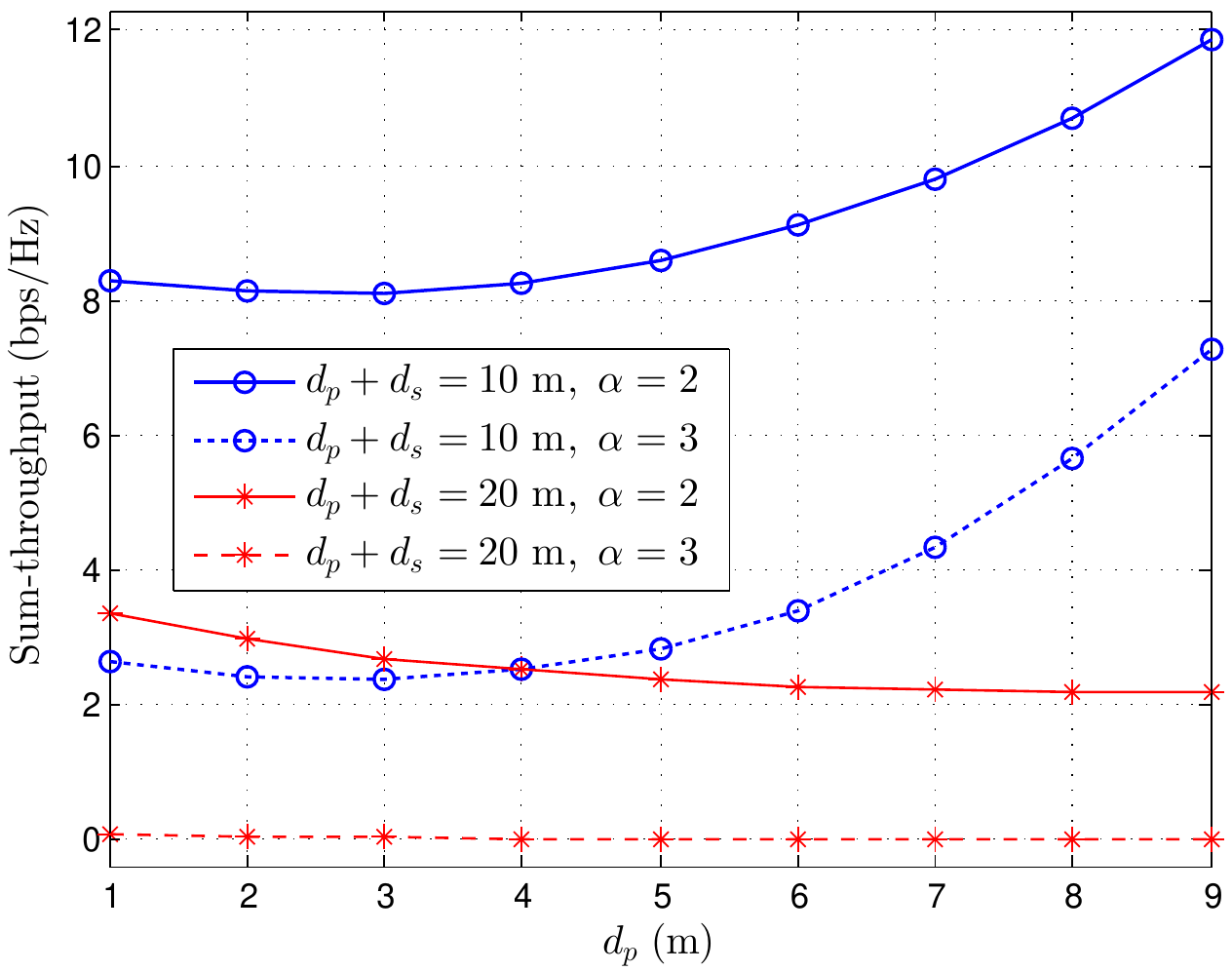}
   \caption{ Sum-throughput vs. $d_p$, where $K\!=\!4$, and $P_{\max}\!=\!30$ dBm.}
   \label{fig:ICTp2}
\end{figure}
Fig. \ref{fig:IddTi} shows the average iteration time vs. the user number $K$ with $d_p\!=\!d_s\!=\!5$ m and $\alpha\!=\!3$. It turns out that the proposed fast algorithm can greatly reduce the time complexity.

Fig. \ref{fig:ICTp2} plots the sum-throughput vs. $d_p$, where $\alpha\!\in\!\{2,3\}$, the PS-sink distance $d_{ps}=d_p+d_s\!\in\!\{10,20\}$ meters and $K\!=\!4$. 
As intuition suggests, the sum-throughput decreases as $\alpha$ or $d_{ps}$ increases.
However, with a given reasonable $d_{ps}$, an interesting thing is that the users should be more close to the sink than PS to achieve high sum-throughput if the PS-sink distance is short; Otherwise, the users should move close to the PS.
This can be interpreted as due to the large-scale fading and energy beamforming transmission at PS.
Hence, there is a tradeoff between the sum-throughput and the node deployment in a practical scenario.

\section{Conclusion\label{Sec6}}
In this letter, we consider the optimal design for a wireless powered communication network. The sum-throughput is maximized by joint time allocation and beamforming. The semidefinite relaxation technique is applied to resolve the nonconvexity issue, and its tightness is proved. A fast algorithm is further proposed to substantially reduce the time complexity. Simulation results demonstrate the effectiveness of the proposed algorithm.

\bibliographystyle{IEEEtran}

\end{document}